\documentclass[letter,romanappendices,onecolumn,draft]{ieeeconf}
\let\proof\relax   

\usepackage{amsthm,xpatch}
\usepackage{amsmath,amsfonts}
\usepackage{cite}
\usepackage{amssymb}
\usepackage{dsfont}
\usepackage{graphicx, subfigure}
\usepackage{color}
\usepackage{breqn}
\usepackage{mathtools}
\usepackage{bbm}
\usepackage{latexsym}
\usepackage[ruled, linesnumbered]{algorithm2e}
\usepackage{accents}
\usepackage{tikz}
\usepackage[mathscr]{euscript}
\usepackage{bm}
\usepackage{comment}

\usepackage[inline]{trackchanges}

\addeditor{}

\newtheorem{lemma}{Lemma}
\newtheorem{theorem}{Theorem}
\newtheorem{remark}{Remark}

\newtheorem{example}{Example}

\newcommand*{\transpose}{%
  {\mathpalette\@transpose{}}%
}

\IEEEoverridecommandlockouts

\begin{document}

\newcommand{\SB}[3]{
\sum_{#2 \in #1}\biggl|\overline{X}_{#2}\biggr| #3
\biggl|\bigcap_{#2 \notin #1}\overline{X}_{#2}\biggr|
}

\newcommand{\Mod}[1]{\ (\textup{mod}\ #1)}

\newcommand{\overbar}[1]{\mkern 0mu\overline{\mkern-0mu#1\mkern-8.5mu}\mkern 6mu}

\makeatletter
\newcommand*\nss[3]{%
  \begingroup
  \setbox0\hbox{$\m@th\scriptstyle\cramped{#2}$}%
  \setbox2\hbox{$\m@th\scriptstyle#3$}%
  \dimen@=\fontdimen8\textfont3
  \multiply\dimen@ by 4             
  \advance \dimen@ by \ht0
  \advance \dimen@ by -\fontdimen17\textfont2
  \@tempdima=\fontdimen5\textfont2  
  \multiply\@tempdima by 4
  \divide  \@tempdima by 5          
  \ifdim\dimen@<\@tempdima
    \ht0=0pt                        
    \@tempdima=\fontdimen5\textfont2
    \divide\@tempdima by 4          
    \advance \dimen@ by -\@tempdima 
    \ifdim\dimen@>0pt
      \@tempdima=\dp2
      \advance\@tempdima by \dimen@
      \dp2=\@tempdima
    \fi
  \fi
  #1_{\box0}^{\box2}%
  \endgroup
  }
\makeatother

\makeatletter
\renewenvironment{proof}[1][\proofname]{\par
  \pushQED{\qed}%
  \normalfont \topsep6\p@\@plus6\p@\relax
  \trivlist
  \item[\hskip\labelsep
        \itshape
    #1\@addpunct{:}]\ignorespaces
}{%
  \popQED\endtrivlist\@endpefalse
}
\makeatother

\makeatletter
\newsavebox\myboxA
\newsavebox\myboxB
\newlength\mylenA

\newcommand*\xoverline[2][0.75]{%
    \sbox{\myboxA}{$\m@th#2$}%
    \setbox\myboxB\null
    \ht\myboxB=\ht\myboxA%
    \dp\myboxB=\dp\myboxA%
    \wd\myboxB=#1\wd\myboxA
    \sbox\myboxB{$\m@th\overline{\copy\myboxB}$}
    \setlength\mylenA{\the\wd\myboxA}
    \addtolength\mylenA{-\the\wd\myboxB}%
    \ifdim\wd\myboxB<\wd\myboxA%
       \rlap{\hskip 0.5\mylenA\usebox\myboxB}{\usebox\myboxA}%
    \else
        \hskip -0.5\mylenA\rlap{\usebox\myboxA}{\hskip 0.5\mylenA\usebox\myboxB}%
    \fi}
\makeatother

\xpatchcmd{\proof}{\hskip\labelsep}{\hskip3.75\labelsep}{}{}

\pagestyle{plain}

\title{\fontsize{21}{28}\selectfont Private Linear Transformation: The Joint Privacy Case}

\author{Nahid Esmati, Anoosheh Heidarzadeh, and Alex Sprintson\thanks{The authors are with the Department of Electrical and Computer Engineering, Texas A\&M University, College Station, TX 77843 USA (E-mail: \{nahid, anoosheh, spalex\}@tamu.edu).}}


%


\maketitle 

\thispagestyle{plain}

\begin{abstract}

We introduce 
the problem of Private Linear Transformation (PLT). This problem includes a single (or multiple) remote server(s) storing (identical copies of) $K$ messages and a user who wants to compute $L$ linear combinations of a $D$-subset of these messages by downloading the minimum amount of information from the server(s) while protecting the privacy of the entire set of $D$ messages. This problem generalizes the Private Information Retrieval and Private Linear Computation problems. In this work, we focus on the single-server case. For the setting in which the coefficient matrix of the required $L$ linear combinations generates a Maximum Distance Separable (MDS) code, we characterize the capacity---defined as the supremum of all achievable download rates, for all parameters $K, D, L$. In addition, we present lower and/or upper bounds on the capacity for the settings with non-MDS coefficient matrices and the settings with a prior side information.  
\end{abstract}



\section{introduction}
This work 
introduces the problem of 
\emph{Private Linear Transformation (PLT)}. This problem includes a single (or multiple non-colluding or with limited collusion capability) remote server(s) storing (identical copies of) a dataset consisting of $K$ data items; and a user who is interested in computing $L$ linear combinations of a $D$-subset of data items. The goal of the user is to perform the computation privately so that the identities of the data items required for the computation are protected (to some degree) from the server(s), while minimizing the total amount of information being downloaded from the server(s).
The PLT problem generalizes the problems of Private Information Retrieval (PIR) (see, e.g.,~\cite{CGKS1995,SJ2017,BU17,BU2018,SJ2018No2,TSC2018,CWJ2018}) and Private Linear Computation (PLC) (see, e.g.,~\cite{SJ2018,MM2018,TM2019No2}), which have recently received a significant attention from the information and coding theory community.
In particular, PLT reduces to PIR or PLC when $L=D$ or $L=1$, respectively. In this work, we focus on the single-server setting of PLT. 

This problem setup appears in several practical scenarios such as linear transformation for dimensionality reduction in Machine Learning (ML), see, e.g.,~\cite{CG2015} and references therein. 
Consider a dataset with $N$ data samples, each with $K$ attributes. 
Consider a user that wishes to implement an ML  algorithm on a subset of $D$ selected attributes, while protecting the privacy of the selected attributes.  
When $D$ is large, the $D$-dimensional feature space is typically mapped onto a new subspace of lower dimension, say, $L$, and the ML algorithm operates on the new $L$-dimensional subspace instead. 
A commonly-used technique for dimensionality reduction is \emph{linear transformation}, where an $L\times D$ matrix is multiplied by the $D\times N$ data submatrix (the submatrix of the original $K\times N$ data matrix restricted to the $D$ selected attributes). 
Thinking of the rows of the $K\times N$ data matrix as the $K$ messages, the labels of the $D$ selected attributes as the identities of the $D$ messages in the support set of the required linear combinations, and the $L\times D$ matrix used for transformation as the coefficient matrix of the required linear combinations, this scenario matches the setup of the PLT problem.

A natural approach for PLT is to privately retrieve the items required for the computation using a PIR scheme, and then compute the required linear combinations locally. 
As was shown in~\cite{KGHERS2020,HKS2019Journal,HKRS2019,KKHS32019,HKGRS2018,LG2018,HKS2018,HKS2019}, leveraging a prior side information about the dataset, in the single-server setting, the user can retrieve a single or multiple data items privately with a much lower download cost than the trivial scheme of downloading the entire dataset. 
(The advantages of side information in multi-server PIR were also studied in~\cite{T2017,WBU2018,WBU2018No2,CWJ2017,SSM2018,KKHS22019,KKHS12019}.) 
However, when there is no side information, a PIR-based PLT approach is extremely expensive since all items in the dataset must be downloaded in order to achieve privacy~\cite{CGKS1995}. 

Another approach for PLT is to privately compute the required linear combinations separately via applying a PLC scheme multiple times. 
In~\cite{HS2019PC,HS2020}, it was shown that single-server PLC can be performed more efficiently than single-server PIR in terms of the download cost, regardless of whether the user has any side information or not. This suggests that a PLC-based PLT scheme can outperform a PIR-based PLT scheme; however, a PLC-based approach may still lead to an unnecessary overhead due to the excessive redundancy in the information being downloaded. This implies the need for novel PLT schemes with optimal download rate.

\subsection{Main Contributions}

In this paper, we study the problem of single-server PLT under a strong notion of privacy, which we call \emph{joint privacy}, 
where the user wants to hide the identities of all items required for the computation jointly, and leak no information about the correlation between them. 
We refer to this problem as \emph{PLT with Joint Privacy}, or \emph{JPLT} for short. 

We focus on the setting in which the coefficient matrix of the required linear combinations generates a Maximum Distance Separable (MDS) code. 
The MDS coefficient matrices are motivated by the application of \emph{random linear transformation} for dimensionality reduction (see, e.g.,~\cite{BM2001}), where a random $L\times D$ matrix is used for transformation. 
Note that 
an $L\times D$ matrix whose entries are randomly chosen from a sufficiently large alphabet is MDS with high probability. 
For this setting, we characterize the capacity of JPLT, where the capacity is defined as the supremum of download rates over all JPLT schemes. 
We prove the converse by using a mix of linear-algebraic and information-theoretic arguments, relying on a necessary condition for any JPLT scheme. 
We propose an achievability scheme, termed \emph{Specialized MDS Code protocol}, which leverages the idea of extending an MDS code. 
In addition, we briefly discuss the settings with non-MDS coefficient matrices and the settings with a prior side information; and present lower and/or upper bounds on the capacity of these settings.

\section{Problem Setup}\label{sec:SN}
Throughout, we denote random variables and their realizations by bold-face symbols and regular symbols, respectively. 

Let $\mathbb{F}_p$ be a finite field of order $p$, and let $\mathbb{F}_{q}$ be an extension field of $\mathbb{F}_p$. 
Let $K,D,L$ be positive integers such that ${L\leq D\leq K}$, and let $\mathcal{K}$ denote the set of integers $\{1,...,K\}$. Let $\mathscr{W}$ be the set of all $D$-subsets $\mathrm{W}$ of $\mathcal{K}$, and 
let $\mathscr{V}$ be the set of all $L\times D$ matrices $\mathrm{V}$ (with entries from $\mathbb{F}_p$) that are \emph{Maximum Distance Separable (MDS)}, i.e., every $L\times L$ submatrix of $\mathrm{V}$ is invertible. 

Consider a server that stores $K$ messages ${X_1,\dots,X_K}$, where $X_i\in \mathbb{F}_q$ for $i\in \mathcal{K}$. 
Let ${\mathrm{X}\triangleq [X_1,\dots,X_K]^{\mathsf{T}}}$.
For every ${\mathrm{S}\subset \mathcal{K}}$, denote by $\mathrm{X}_{\mathrm{S}}$ the vector $\mathrm{X}$ restricted to its components indexed by $\mathrm{S}$, i.e., $\mathrm{X}_{\mathrm{S}} = [X_{i_1},\dots,X_{i_s}]^{\mathsf{T}}$, where ${\mathrm{S} = \{i_1,\dots,i_s\}}$. 
Assume that $\mathbf{X}_1,\dots,\mathbf{X}_K$ are independently and uniformly distributed over $\mathbb{F}_{q}$. 
Thus, ${H(\mathbf{X}_{\mathrm{S}})= |\mathrm{S}| \theta}$ for every ${\mathrm{S}\subset \mathcal{K}}$, where $|\mathrm{S}|$ denotes the size of $\mathrm{S}$, and $\theta\triangleq \log_2 q$. 
Note that $H(\mathbf{X})=K\theta$. 

Consider a user who wishes to compute $L$ linear combinations $\mathrm{v}^{\mathsf{T}}_1 \mathrm{X}_{\mathrm{W}},\dots,\mathrm{v}^{\mathsf{T}}_L \mathrm{X}_{\mathrm{W}}$, collectively denoted by the vector $\mathrm{Z}^{[\mathrm{W},\mathrm{V}]}\triangleq \mathrm{V}\mathrm{X}_{\mathrm{W}}$, where ${\mathrm{W}\in \mathscr{W}}$ and ${\mathrm{V}= [\mathrm{v}_{1},\dots,\mathrm{v}_{L}]^{\mathsf{T}}\in \mathscr{V}}$. 
Note that $H(\mathbf{Z}^{[\mathrm{W},\mathrm{V}]})=L\theta$. 
We refer to $\mathrm{Z}^{[\mathrm{W},\mathrm{V}]}$ as the \emph{demand}, $\mathrm{W}$ as the \emph{support index set of the demand}, $\mathrm{V}$ as the \emph{coefficient matrix of the demand}, $D$ as the \emph{support size of the demand}, and $L$ as the \emph{dimension of the demand}.

In this work, we assume that (i) $\mathbf{W}$, $\mathbf{V}$, and $\mathbf{X}$ are independent; (ii) $\mathbf{W}$ is uniformly distributed over all ${\mathrm{W}\in \mathscr{W}}$; 
(iii) $\mathbf{V}$ is uniformly distributed over all ${\mathrm{V}\in\mathscr{V}}$; 
and (iv) the parameters $D$ and $L$, and the joint distribution of $\mathbf{W}$ and $\mathbf{V}$ are initially known by the server, whereas the server does not initially know the realizations $\mathrm{W}$ and $\mathrm{V}$. 

Given $\mathrm{W}$ and $\mathrm{V}$, the user generates a query $\mathrm{Q}^{[\mathrm{W},\mathrm{V}]}$, simply denoted by $\mathrm{Q}$, and sends it to the server. 
The query $\mathrm{Q}$ is a function of $\mathrm{W}$, $\mathrm{V}$, and potentially a random key $\mathrm{R}$ (independent of $\mathrm{W}$, $\mathrm{V}$, and $\mathrm{X}$) that is generated by the user and is initially unknown to the server. 
That is, $H(\mathbf{Q}|\mathbf{W},\mathbf{V},\mathbf{R})=0$, where $\mathbf{Q}^{[\mathbf{W},\mathbf{V}]}$ is denoted by $\mathbf{Q}$. 

Given the query $\mathrm{Q}$, every $D$-subset of message indices must be equally likely to be the demand's support index set, i.e., for every $\mathrm{W}^{*} \in \mathscr{W}$, it must hold that
\begin{equation*}
\Pr (\mathbf{W}=\mathrm{W}^{*}|\mathbf{Q}=\mathrm{Q})=\Pr(\mathbf{W}=\mathrm{W}^{*}). 
\end{equation*} This condition, which we refer to as the \emph{joint privacy condition}, was previously considered for PIR and PLC (see, e.g.,~\cite{BU2018,HKGRS2018,HS2020}). 
Note that, for this type of privacy, it is not required that the user's query protects the privacy of the demand's coefficient matrix from the server. This is inspired by 
several real-world scenarios. For example, protecting the privacy of the selected attributes in the application of random linear transformation prevents the server (to some extent) from learning the inner working of the user's ML algorithm.

Upon receiving the query $\mathrm{Q}$, the server generates an answer $\mathrm{A}^{[\mathrm{W},\mathrm{V}]}$, simply denoted by $\mathrm{A}$, and sends it back to the user. 
The answer $\mathrm{A}$ is a deterministic function of $\mathrm{Q}$ and $\mathrm{X}$.
That is, $H(\mathbf{A}|\mathbf{Q},\mathbf{X})=0$, where $\mathbf{A}^{[\mathbf{W},\mathbf{V}]}$ is denoted by $\mathbf{A}$. 
The answer $\mathrm{A}$, the query $\mathrm{Q}$, and the realizations $\mathrm{W}, \mathrm{V}$ must collectively enable the user to retrieve the demand $\mathrm{Z}^{[\mathrm{W},\mathrm{V}]}$, i.e., 
$H(\mathbf{Z}| \mathbf{A},\mathbf{Q}, \mathbf{W},\mathbf{V})=0,$ where $\mathbf{Z}^{[\mathbf{W},\mathbf{V}]}$ is denoted by $\mathbf{Z}$. 
We refer to this condition as the \emph{recoverability condition}. 

The problem is to design a protocol for generating a query $\mathrm{Q}^{[\mathrm{W},\mathrm{V}]}$ and the corresponding answer $\mathrm{A}^{[\mathrm{W},\mathrm{V}]}$ 
such that both the joint privacy and recoverability conditions are satisfied.
We refer to this problem as single-server \emph{Private Linear Transformation (PLT) with Joint Privacy}, or \emph{JPLT} for short. 

Following the convention in the PIR and PLC literature, we measure the efficiency of a JPLT protocol by its \emph{rate}---defined as the ratio of the entropy of the demand (i.e., $H(\mathbf{Z})=L\theta$) to the entropy of the answer (i.e., $H(\mathbf{A})$). 
We define the \emph{capacity} of the JPLT setting as the supremum of rates over all JPLT protocols.
Our goal in this work is to characterize the capacity of the JPLT setting.

In~\cite{EHS2021Individual}, we have considered the problem of single-server \emph{PLT with Individual Privacy}, or \emph{IPLT} for short, 
where
every message index must be equally likely to belong to the demand's support index set. This condition was previously considered for PIR and PLC (see, e.g.,~\cite{HKRS2019,KKHS32019,HS2020}). Note that individual privacy is weaker than joint privacy, because not every $D$-subset of message indices needs to be equally likely to be the demand's support index set.   

\section{A Necessary Condition for JPLT Protocols}
The following lemma states a necessary (yet not always sufficient) condition for any JPLT protocol. This result follows immediately from the joint privacy and recoverability conditions, and its proof is omitted for brevity. 

\begin{lemma}\label{lem:NCJPLT}
Given any JPLT protocol, for any $\mathrm{W}^{*}\in\mathscr{W}$, there must exist $\mathrm{V}^{*}\in\mathscr{V}$, such that $H(\mathbf{Z}^{[\mathrm{W}^{*},\mathrm{V}^{*}]}| \mathbf{A}, \mathbf{Q})= 0$. 		
\end{lemma}

When considering \emph{linear} JPLT schemes, i.e., the schemes in which the answer consists of only \emph{linear} combinations of the messages, the necessary condition provided by Lemma~\ref{lem:NCJPLT} can be interpreted in the language of coding theory as follows. 
The coefficient matrix of the linear combinations corresponding to the answer must generate a (linear) code of length $K$ that, when punctured at any $K-D$ coordinates, contains $L$ codewords that are MDS, i.e., they generate a $[D,L]$ MDS code. (Puncturing a (linear) code at a coordinate is performed by deleting the column pertaining to that coordinate from the generator matrix of the code.)
A code satisfying this condition is, however, not guaranteed to yield a JPLT scheme. 
A sufficient (but not necessary) condition is that 
the codes resulting from puncturing at any $K-D$ coordinates contain the \emph{same number of groups} of $L$ MDS codewords. 
Thus, designing a linear JPLT scheme with \emph{maximum rate} reduces to constructing such a linear code with \emph{minimum dimension}. 
This sufficient condition is, however, more combinatorial in nature, and the necessary condition provided by Lemma~\ref{lem:NCJPLT} proves more useful when deriving an information-theoretic converse bound.

\section{Main Results}
In this section, we summarize our main results for JPLT. 

\begin{theorem}\label{thm:JPLT}
For the JPLT setting with $K$ messages, demand's support size $D$, and demand's dimension $L$, the capacity is given by $L/(K-D+L)$. 	
\end{theorem}

The proof of converse is based on information theoretic arguments relying mostly on the necessary condition for JPLT protocols---provided by Lemma~\ref{lem:NCJPLT}. 
The converse bound naturally serves as an upper bound on the rate of any JPLT protocol. 
We prove the achievability by designing a linear JPLT protocol, termed the \emph{Specialized MDS Code protocol}, that achieves the converse bound. 
This protocol generalizes those in~\cite{HKGRS2018} and~\cite{HS2019PC} for single-server PIR and PLC (without SI) with joint privacy, and is based on the idea of extending the MDS code generated by the coefficient matrix of the demand. 
In particular, when the coefficient matrix of the demand generates a Generalized Reed-Solomon (GRS) code, 
we give an explicit construction of a GRS code that contains a specific collection of codewords---specified by the demand's support index set and coefficient matrix. 

\begin{remark}\label{rem:PLT}
\emph{The result of Theorem~\ref{thm:JPLT} shows that, when there is only a single server and there is no side information available at the user, JPLT can be performed more efficiently than using either of the following two approaches: 
(i) retrieving the messages required for computation using a multi-message PIR scheme~\cite{HKGRS2018} and computing the required linear combinations locally, or 
(ii) computing each of the required linear combinations separately via applying a PLC scheme~\cite{HS2019PC}. 
More specifically, the optimal rate for the approach (i) or (ii) is $L/K$ or $1/(K-D+1)$, respectively, whereas an optimal JPLT scheme achieves the rate $L/(K-D+L)$.}
\end{remark}

\begin{remark}\label{rem:JPLT1}
\emph{In~\cite{HS2019PC}, it was shown that the rate ${1/(K-D+1)}$ is achievable for PLC (without SI) with joint privacy, but no converse result was presented. 
The result of Theorem~\ref{thm:JPLT} for ${L=1}$ proves the optimality of this rate.   
For $L=D$, the problem reduces to PIR (without SI) when joint privacy is required, and as was shown in~\cite{HKGRS2018}, an optimal solution for this case is to download the entire dataset.}	
\end{remark}

\begin{remark}\label{rem:JPLT2}
\emph{Theorem~\ref{thm:JPLT} can be extended to 
JPLT with Side Information (SI). 
Two types of SI were previously studied for PIR and PLC: \emph{Uncoded SI (USI)} (see, e.g.,~\cite{KGHERS2020}), and \emph{Coded SI (CSI)} (see, e.g.,~\cite{HKS2019Journal}). 
In the case of USI, the user initially knows a subset of $M$ messages, whereas in the case of CSI, 
the user initially knows $L$ MDS coded combinations of $M$ messages. In both cases, the identities of these $M$ messages are initially unknown by the server. 
Using similar techniques 
as in this work, we can show that 
the capacity of JPLT with USI is given by ${L/(K-D-M+L)}$ when the identities of the messages in the support set of the demand and those in the support set of the side information need to be protected jointly. 
Similarly, we can show that the capacity of JPLT with CSI (under the same privacy condition) does not change, provided that the coefficient matrix of the side information and that of the demand form an MDS matrix when concatenated horizontally. 
}
\end{remark}

\begin{remark}\label{rem:JPLT3}
\emph{The JPLT setting can also be extended by relaxing the MDS assumption on the coefficient matrix of the demand. 
In particular, using a JPLT scheme similar to the one for the case of MDS matrices, we can achieve the rate $L/(K-D+L)$ when the demand's coefficient matrix is randomly chosen from the set of all \emph{full-rank} (but not necessarily MDS) matrices. 
This rate, however, may not be optimal, and the converse is still open. 
Furthermore, when the demand's coefficient matrix is randomly chosen from the set of all full-rank matrices with \emph{nonzero columns}, using the same proof technique as in the case of MDS matrices, we can show that the rate of any JPLT scheme is upper bound by $L/(K-D+L)$. 
However, the achievability of this rate upper bound remains unknown in general.}
\end{remark}

\section{Proof of Converse}\label{sec:JPLT-Conv}



\begin{lemma}\label{lem:JPLT-Conv}
The rate of any JPLT protocol for $K$ messages, demand's support size $D$, and demand's dimension $L$ is upper bounded by $L/(K-D+L)$.
\end{lemma}

\begin{proof}
Consider an arbitrary JPLT protocol that generates the query-answer pair $(\mathrm{Q}^{[\mathrm{W},\mathrm{V}]},\mathrm{A}^{[\mathrm{W},\mathrm{V}]})$ for any given $(\mathrm{W},\mathrm{V})$. 
To show that the rate of this protocol is upper bounded by $L/(K-D+L)$, we need to show that ${H(\mathbf{A})\geq (K-D+L)\theta}$, where $\mathbf{A}$ denotes $\mathbf{A}^{[\mathbf{W},\mathbf{V}]}$, and $\theta$ is the entropy of a message. 
Let $T\triangleq K-D+1$. 
For every $1\leq i\leq T$, let $\mathrm{W}_i \triangleq {\{i,i+1,\dots,i+D-1\}}$. 
Note that ${\mathrm{W}_1,\dots,\mathrm{W}_T\in \mathscr{W}}$.
By Lemma~\ref{lem:NCJPLT}, for any ${1\leq i\leq T}$, there exists ${\mathrm{V}_i\in \mathscr{V}}$ such that ${H(\mathbf{Z}_i |\mathbf{A},\mathbf{Q}) = 0}$, where ${\mathbf{Z}_i\triangleq \mathbf{Z}^{[\mathrm{W}_i,\mathrm{V}_i]}}$. (Note that $\mathrm{V}_i$ is an MDS matrix.)
This readily implies that $H(\mathbf{Z}_1,\dots,\mathbf{Z}_{T}|\mathbf{A},\mathbf{Q})=0$ since $H(\mathbf{Z}_1,\dots,\mathbf{Z}_{T}|\mathbf{A},\mathbf{Q})\leq \sum_{i=1}^{T} H(\mathbf{Z}_i|\mathbf{A},\mathbf{Q}) = 0$. 
Thus, 
\begin{align}
    H(\mathbf{A})&\geq H(\mathbf{A}|\mathbf{Q}) +H(\mathbf{Z}_1,\dots,\mathbf{Z}_{T}|\mathbf
    {Q},\mathbf{A}) \label{eq:1}\\
    &= H(\mathbf{Z}_1,\dots,\mathbf{Z}_{T}|\mathbf{Q})
    +H(\mathbf{A}|\mathbf{Q},\mathbf{Z}_1,\dots,\mathbf{Z}_{T}) \label{eq:2}\\
    &\geq H(\mathbf{Z}_1,\dots,\mathbf{Z}_{T}), \label{eq:3}
\end{align} 
where~\eqref{eq:1} holds because $H(\mathbf{Z}_1,\dots,\mathbf{Z}_{T}|\mathbf{A},\mathbf{Q})=0$, as shown earlier; \eqref{eq:2} follows from the chain rule of conditional entropy;
and \eqref{eq:3} holds because 
(i) $\mathbf{Z}_i$'s are independent from $\mathbf{Q}$, noting that $\mathbf{Z}_i$'s only depend on $\mathbf{X}$, and $\mathbf{Q}$ is independent of $\mathbf{X}$, and 
(ii) $H(\mathbf{A}|\mathbf{Q},\mathbf{Z}_1,\dots,\mathbf{Z}_{T})\geq 0$. 

To lower bound $H(\mathbf{Z}_1,\dots,\mathbf{Z}_{T})$, we proceed as follows. 
By the chain rule of entropy, we have $H(\mathbf{Z}_1,\dots,\mathbf{Z}_{T})=H(\mathbf{Z}_1)+\sum_{1<i\leq T} H(\mathbf{Z}_{i}|\mathbf{Z}_{1},\dots,\mathbf{Z}_{i-1})$. 
Let $\mathbf{Z}_{i,1},\dots,\mathbf{Z}_{i,L}$ be the $L$ components of the vector $\mathbf{Z}_i$, i.e., $\mathbf{Z}_{i,l}\triangleq \mathrm{v}^{\mathsf{T}}_{i,l} \mathbf{X}_{\mathrm{W}_i}$, where $\mathrm{v}^{\mathsf{T}}_{i,l}$ is the $l$th row of $\mathrm{V}_i$. 
Note that $\mathbf{Z}_i$ consists of $L$ components $\mathbf{Z}_{i,1},\dots,\mathbf{Z}_{i,L}$, and 
these components are independent because their corresponding coefficient vectors $\mathrm{v}_{i,1},\dots,\mathrm{v}_{i,L}$ are linearly independent. 
Moreover, $\mathbf{Z}_{i,1},\dots,\mathbf{Z}_{i,L}$ are uniform over $\mathbb{F}_q$, i.e., $H(\mathbf{Z}_{i,l})=\theta$ for $l\in \{1,\dots,L\}$. 
Thus, $H(\mathbf{Z}_i)=H(\mathbf{Z}_{i,1},\dots,\mathbf{Z}_{i,L})=L\theta$, particularly, ${H(\mathbf{Z}_1)=L\theta}$. 
Obviously, $\mathbf{X}_{i-D+1}$ belongs to the support set of $\mathbf{Z}_{i,l}$ for some ${l\in \{1,\dots,L\}}$. Otherwise, $\mathrm{V}_i$ contains an all-zero column, which is a contradiction. Moreover, $\mathbf{X}_{i-D+1}$ does not belong to the support set of any of the components of $\mathbf{Z}_j$ for any $j<i$ (by construction). 
This implies that $\mathbf{Z}_i$ contains at least one component, namely, $\mathbf{Z}_{i,l}$, that cannot be written as a linear combination of the components in $\mathbf{Z}_1,\dots,\mathbf{Z}_{i-1}$. 
Thus, $\mathbf{Z}_{i,l}$ is independent of $\mathbf{Z}_1,\dots,\mathbf{Z}_{i-1}$. 
This further implies that $H(\mathbf{Z}_{i}|\mathbf{Z}_{1},\dots,\mathbf{Z}_{i-1})\geq H(\mathbf{Z}_{i,l})=\theta$, and consequently, $\sum_{1<i\leq T} H(\mathbf{Z}_{i}|\mathbf{Z}_{1},\dots,\mathbf{Z}_{i-1})\geq {(T-1)\theta}$. 
Thus,
\begin{equation}\label{eq:4}
H(\mathbf{Z}_1,\dots,\mathbf{Z}_T)\geq L\theta + (T-1)\theta = (K-D+L)\theta.    
\end{equation} Combining~\eqref{eq:3} and~\eqref{eq:4}, we have $H(\mathbf{A})\geq (K-D+L)\theta$.
\end{proof}

\section{Achievability Scheme}\label{sec:JPLT-Ach}

In this section, we propose a JPLT protocol, termed \emph{Specialized MDS Code Protocol}, that achieves the rate $L/(K-D+L)$.  
This protocol consists of three steps.\vspace{0.125cm} 

\textbf{Step 1:} Given the demand support index set $\mathrm{W}$ and the demand coefficient matrix $\mathrm{V}=[\mathrm{v}_1,\dots,\mathrm{v}_l]^{\mathsf{T}}$, the user constructs a query $\mathrm{Q}^{[\mathrm{W},\mathrm{V}]}$ in the form of a matrix $\mathrm{G}$, such that the user's query, i.e., the matrix $\mathrm{G}$, and the server's corresponding answer $\mathrm{A}^{[\mathrm{W},\mathrm{V}]}$, i.e., the vector $\mathrm{G}\mathrm{X}$, 
satisfy the recoverability and privacy conditions.

To satisfy privacy, it is required that, for any index set $\mathrm{W}^{*}\in\mathscr{W}$, the code generated by $\mathrm{G}$ contains $L$ codewords whose support index sets are some subsets of $\mathrm{W}^{*}$, and the coordinates of these codewords (indexed by $\mathrm{W}^{*}$) form an MDS matrix $\mathrm{V}^{*}\in \mathscr{V}$. 
By the properties of MDS codes~\cite{R2006}, it is easy to verify that the generator matrix of any $[K,K-D+L]$ MDS code satisfies this requirement. 
Any such matrix, however, is not guaranteed to satisfy the recoverability condition.
For satisfying recoverability, it is required that $\mathrm{G}$, as a generator matrix, generates a code that contains $L$ codewords with the support $\mathrm{W}$, and the coordinates of these codewords (indexed by $\mathrm{W}$) must conform to the coefficient matrix $\mathrm{V}$. 
To construct a matrix $\mathrm{G}$ that satisfies these requirements, the user proceeds as follows.

First, the user constructs the parity-check matrix $\Lambda$ of the $[D,L]$ MDS code generated by $\mathrm{V}$. 
Since $\mathrm{V}$ is an MDS matrix, then $\Lambda$ generates a $[D,D-L]$ MDS code (i.e., the dual of the MDS code generated by $\mathrm{V}$). 

The user then constructs a $(D-L)\times K$ matrix $\mathrm{H}$ that satisfies the following two conditions: 
\begin{itemize}
    \item[(i)] The submatrix of $\mathrm{H}$ restricted to columns indexed by $\mathrm{W}$ (and all rows) is $\Lambda$, and
    \item[(ii)] The matrix $\mathrm{H}$ is MDS. 
\end{itemize}
Since $\Lambda$ is an MDS matrix, constructing $\mathrm{H}$ reduces to extending the $[D,D-L]$ MDS code generated by $\Lambda$ to a $[K,D-L]$ MDS code. (Extending a code is performed by adding new columns to the generator matrix of the code.)
Next, the user constructs a $(K-D+L)\times K$ matrix $\mathrm{G}$ that generates the MDS code defined by the parity-check matrix $\mathrm{H}$. 
(Since $\mathrm{H}$ generates a $[K,D-L]$ MDS code, $\mathrm{H}$ is the parity-check matrix of a $[K,K-D+L]$ MDS code.)
The user sends the matrix $\mathrm{G}$ as the query $\mathrm{Q}^{[\mathrm{W},\mathrm{V}]}$ to the server. 

In the following, we describe how to explicitly construct the matrix $\mathrm{G}$ when the coefficient matrix $\mathrm{V}$ generates a GRS code, i.e., the entry $(i,j)$ of $\mathrm{V}$ is given by $\mathrm{V}_{i,j} \triangleq \nu_{j} \omega_{j}^{i-1}$, where $\nu_{1},\dots,\nu_{D}$ are $D$ elements from $\mathbb{F}_p\setminus \{0\}$, and $\omega_{1},\dots,\omega_{D}$ are $D$ distinct elements from $\mathbb{F}_p$. 
The parameters $\nu_{1},\dots,\nu_{D}$ and $\omega_{1},\dots,\omega_{D}$ are the multipliers and the evaluation points of the GRS code generated by $\mathrm{V}$, respectively. 
Since the dual of a GRS code is a GRS code~\cite{R2006}, the parity-check matrix $\Lambda$ of the GRS code generated by $\mathrm{V}$ is a ${(D-L)\times D}$ matrix whose entry $(i,j)$ is given by ${\Lambda_{i,j} \triangleq \lambda_{j} \omega_{j}^{i-1}}$, where ${\lambda_{j} \triangleq \nu_{j}^{-1}\prod_{k\in \{1,\dots,D\}\setminus \{j\}} (\omega_{j}-\omega_{k})^{-1}}$.
Note that $\lambda_1,\dots,\lambda_D$ are nonzero.  
Extending the $(D-L)\times D$ matrix $\Lambda$ to a $(D-L)\times K$ matrix $\mathrm{H}$---satisfying the conditions (i) and (ii) specified earlier---is performed as follows. 

Let $\mathrm{W}=\{i_1,\dots,i_D\}$ and $\mathcal{K}\setminus \mathrm{W} = \{i_{D+1},\dots,i_K\}$, and 
let $\pi$ be a permutation on $\mathcal{K}$ such that $\pi(j) = i_j$. 
Let $\lambda_{D+1},\dots,\lambda_{K}$ be $K-D$ elements  chosen randomly (with replacement) from $\mathbb{F}_p\setminus \{0\}$, and 
let $\omega_{D+1},\dots,\omega_K$ be $K-D$ elements chosen randomly (without replacement) from $\mathbb{F}_p\setminus \{\omega_1,\dots,\omega_D\}$.
For every ${j\in \{1,\dots,D\}}$, let the $\pi(j)$th column of $\mathrm{H}$ be the $j$th column of $\Lambda$, and for every $j\in \mathcal{K}\setminus \{1,\dots,D\}$, let the $\pi(j)$th column of $\mathrm{H}$ be $[\lambda_j,\lambda_{j}\omega_j,\dots,\lambda_j\omega_j^{D-L-1}]^{\mathsf{T}}$. 
Since $\mathrm{H}$ is the parity-check matrix of a $[K,K-D+L]$ GRS code, the generator matrix of this code, $\mathrm{G}$, can be simply constructed by taking the $\pi(j)$th column of $\mathrm{G}$ to be $[\alpha_j,\alpha_j\omega_j,\dots,\alpha_j\omega_j^{K-D+L-1}]^{\mathsf{T}}$, where $\alpha_j\triangleq \lambda_{j}^{-1}\prod_{k\in \mathcal{K}\setminus \{j\}} (\omega_{j}-\omega_{k})^{-1}$.
The parameters $\{\alpha_{j}\}_{j\in \mathcal{K}}$ and $\{\omega_{j}\}_{j\in \mathcal{K}}$ are the multipliers and the evaluation points of the GRS code generated by $\mathrm{G}$, respectively.\vspace{0.125cm} 

\textbf{Step 2:} Given the query $\mathrm{Q}^{[\mathrm{W},\mathrm{V}]}$, i.e., the matrix $\mathrm{G}$, the server computes $\mathrm{y}\triangleq \mathrm{G}\mathrm{X}$, and sends the vector $\mathrm{y}$ back to the user as the answer $\mathrm{A}^{[\mathrm{W},\mathrm{V}]}$. In particular, when $\mathrm{V}$ generates a GRS code, the $i$th entry of the vector $\mathrm{y} =  [y_1,\dots,y_{K-D+L}]^{\mathsf{T}}$ is given by $y_i = \sum_{j\in{\mathcal{K}}}\alpha_{j}\omega_j^{i-1}X_{j}$.\vspace{0.125cm}  

\textbf{Step 3:} Upon receiving the answer $\mathrm{A}^{[\mathrm{W},\mathrm{V}]}$, i.e., the vector $\mathrm{y}$, 
the user constructs a matrix $[\tilde{\mathrm{G}},\tilde{\mathrm{y}}]$ by 
performing row operations on the augmented matrix $[\mathrm{G},\mathrm{y}]$, so as to zero out the submatrix formed by the first $L$ rows and the columns indexed by $\mathcal{K}\setminus \mathrm{W}$.
Since the submatrix of $[\tilde{\mathrm{G}},\tilde{\mathrm{y}}]$ formed by the first $L$ rows and the columns indexed by $\mathrm{W}$ (or $\mathcal{K}\setminus\mathrm{W}$) is the matrix $\mathrm{V}$ (or an all-zero matrix), 
the $l$th component of the demand vector $\mathrm{Z}^{[\mathrm{W},\mathrm{V}]}$, i.e., $\mathrm{v}^{\mathsf{T}}_l\mathrm{X}_{\mathrm{W}}$, can be recovered as the $l$th entry of the vector $\tilde{\mathrm{y}}$. 
When $\mathrm{V}$ generates a GRS code, $\mathrm{Z}^{[\mathrm{W},\mathrm{V}]}$ can be recovered from the vector $\mathrm{y}$ as follows. 
First, the user constructs $L$ polynomials $f_1(x),\dots,f_L(x)$, where \[f_{l}(x)\triangleq {x^{l-1}\prod_{j\in \mathcal{K}\setminus\{1,\dots,D\}}(x-\omega_j)}.\] 
Let $\mathrm{c}_{l}\triangleq [c_{l,1},\dots,c_{l,K-D+L}]^{\mathsf{T}}$, where $c_{l,i}$ is the coefficient of the monomial $x^{i-1}$ in the expansion of $f_{l}(x)$. 
The user then recovers 
$\mathrm{v}^{\mathsf{T}}_l\mathrm{X}_{\mathrm{W}}$ for $1\leq l\leq L$ by computing $\mathrm{c}_{l}^{\mathsf{T}}\mathrm{y}$. 

\begin{example}
\normalfont 
Consider  a  scenario  where  the  server has ${K=10}$ messages $\mathrm{X}_1,\dots,\mathrm{X}_{10}\in\mathbb{F}_{11}$, and the user wants to compute ${L=2}$ linear combinations of $D=5$ messages $X_2$, $X_4$, $X_5$, $X_7$, $X_8$, say, $Z_1 = X_2+3X_4+2X_5+X_7+6X_8$ and
$Z_2 = 3X_2+10X_4+7X_5+4X_7+8X_8$. Note that for this example, the demand's support index set $\mathrm{W}=\{2,4,5,7,8\}$, and the demand's coefficient matrix 
\begin{equation*}
\mathrm{\mathrm{V}} = 
\begin{bmatrix}
1 & 3 & 2 & 1 & 6\\
3 & 10 & 7 & 4 & 8\\
\end{bmatrix}.
\end{equation*}
It is easy to verify that $\mathrm{V}$ generates a $[5,2]$ GRS code with the multipliers $\{\nu_{1},\dots,\nu_{5}\}=\{1,3,2,1,6\}$ and the evaluation points $\{\omega_{1},\dots,\omega_{5}\}=\{3,7,9,4,5\}$. Thus, the user can obtain the parity-check matrix $\Lambda$ of this code as 
\begin{equation*}
\mathrm{\Lambda} = 
\begin{bmatrix}
3 & 10 & 8 & 8 & 7\\
9 & 4 & 6 & 10 & 2\\
5 & 6 & 10 & 7 & 10 \\
\end{bmatrix}.
\end{equation*}
Note that $\Lambda$ generates a $[5,3]$ MDS code with the multipliers $\{\lambda_1,\dots,\lambda_5\} = \{3,10,8,8,7\}$ and the evaluation points $\{\omega_1,\dots,\omega_5\} = \{3,7,9,4,5\}$. 
Next, the user extends the ${3\times 5}$ matrix $\Lambda$ to a ${3\times 10}$ matrix $\mathrm{H}$ that satisfies the conditions (i) and (ii) specified in the step 1 of the protocol. 
Suppose the user randomly chooses $6$ additional multipliers $\{\lambda_{6},\dots,\lambda_{10}\}=\{3,5,1,1,4\}$ (from $\mathbb{F}_{11}\setminus \{0\}$) and $6$ additional evaluation points $\{\omega_{6},\dots,\omega_{10}\}=\{6,1,10,2,8\}$ (from $\mathbb{F}_{11}\setminus \{\omega_1,\dots,\omega_{5}\}$). Followed by constructing a permutation $\pi$ as described in the step 1 of the protocol, say, $\{\pi(1),\dots,\pi(10)\}=\{2,4,5,7,8,1,3,6,9,10\}$, the user constructs the matrix $H$ as 
\begin{equation*}
\mathrm{H} = 
\begin{bmatrix}
3 & \mathbf{3} & 5 & \mathbf{10} & \mathbf{8} & 1 & \mathbf{8} & \mathbf{7} & 1 & 4\\
7 & \mathbf{9} & 5 & \mathbf{4} & \mathbf{6} & 10 & \mathbf{10} & \mathbf{2} & 2 & 10\\
9 & \mathbf{5} & 5 & \mathbf{6} & \mathbf{10} & 1 & \mathbf{7} & \mathbf{10} & 4 & 3
\end{bmatrix},
\end{equation*} where the columns indexed by $\pi(1), \pi(2), \pi(3), \pi(4), \pi(5)$ (i.e., $2,4,5,7,8$) correspond to the columns $1,2,3,4,5$ of $\Lambda$, respectively, and the columns indexed by $\pi(6),\pi(7),\pi(8),\pi(9),\pi(10)$ (i.e., $1,3,6,9,10$) correspond to the columns of the generator matrix of a $[5,3]$ GRS code with the multipliers $\{\lambda_6,\dots,\lambda_{10}\}$ and the evaluation points  $\{\omega_6,\dots,\omega_{10}\}$. 
That is, for every $i\in \{6,\dots,10\}$, the $\pi(i)$th column of $\mathrm{H}$ is given by $[\lambda_{i},\lambda_i\omega_i,\lambda_i\omega_i^2]^{\mathsf{T}}$. 
Since $\mathrm{H}$ generates a $[10,3]$ GRS code with the multipliers $\{\lambda_1,\dots,\lambda_{10}\}$ and the evaluation points $\{\omega_1,\dots,\omega_{10}\}$, $\mathrm{H}$ can be thought of as the parity-check matrix of a $[10,7]$ GRS code with the multipliers $\{\alpha_{1},\dots,\alpha_{10}\}=\{10,7,4,5,4,9,2,1,4,4\}$ and the evaluation points $\{\omega_{1},\dots,\omega_{10}\}=\{6,3,1,7,9,10,4,5,2,8\}$. (The process of computing $\alpha_i$'s is explained in the step 1 of the protocol.) The user then obtains the generator matrix $\mathrm{G}$ of this code, 
\begin{equation*}
\mathrm{G} = 
\begin{bmatrix}
9 & 10 & 2 & 7 & 4 & 1 & 5 & 4 & 4 & 4\\
10 & 8 & 2 & 5 & 3 & 10 & 9 & 9 & 8 & 10\\
5 & 2 & 2 & 2 & 5 & 1 & 3 & 1 & 5 & 3 \\
8 & 6 & 2 & 3 & 1 & 10 & 1 & 5 & 10 & 2\\
4 & 7 & 2 & 10 & 9 & 1 & 4 & 3 & 9 & 5\\
2 & 10 & 2 & 4 & 4 & 10 & 5 & 4 & 7 & 7\\
1 & 8 & 2 & 6 & 3 & 1 & 9 & 9 & 3 & 1\\
\end{bmatrix}.
\end{equation*} Then, the user sends the matrix $\mathrm{G}$ as the query to the server. The server then computes the vector ${\mathrm{y}= \mathrm{G}\mathrm{X}}$, and sends it back to the user. 
Next, the user constructs two polynomials $f_1(x)={(x-\omega_6)(x-\omega_7)(x-\omega_8)(x-\omega_9)(x-\omega_{10})}= {(x-6)(x-1)(x-10)(x-2)(x-8)}$ and $f_2(x)=x f_1(x)$ (for more details, see the step 3 of the protocol). It is easy to verify that the coefficient vectors of the polynomials $f_1(x)$ and $f_2(x)$ are given by $\mathrm{c}_1=[8,1,8,9,6,1,0]^{\mathsf{T}}$ and $\mathrm{c}_2=[0,8,1,8,9,6,1]^{\mathsf{T}}$, respectively.
The user then recovers their demand, i.e., $Z_1$ and $Z_2$, by computing 
\begin{align*}
& Z_1 = \mathrm{c}_{1}^{\mathsf{T}}\mathrm{y} = X_{2}+3X_{4}+2X_{5}+X_{7}+6X_{8},\\
& Z_2 = \mathrm{c}_{2}^{\mathsf{T}}\mathrm{y}={3X_{2}+10X_{4}+7X_{5}+4X_{7}+8X_{8}}.
\end{align*} For this example, the rate of the proposed protocol is $\frac{L}{(K-D+L)} = \frac{2}{7}$, whereas the rate of a PIR-based scheme or a PLC-based scheme is $\frac{L}{K} = \frac{2}{10}$ or $\frac{1}{K-D} = \frac{1}{5}$, respectively.    
\end{example}

\begin{lemma}\label{lem:JPLT-Ach}
The Specialized MDS Code protocol is a JPLT protocol, and achieves the rate $L/(K-D+L)$. 
\end{lemma}

\begin{proof}
Since the answer $\mathrm{y}=\mathrm{G}\mathrm{X}$ is a vector of length $K-D+L$, and the entries of this vector are linearly independent coded combinations of the messages $\mathbf{X}_1,\dots,\mathbf{X}_K$ (noting that the matrix $\mathrm{G}$ is full rank), the entropy of the answer is given by ${(K-D+L)\theta}$, where $\theta$ is the entropy of a message. Thus, the rate of this protocol is $L/(K-D+L)$. 

Since $\mathrm{G}$ generates a $[K,K-D+L]$ MDS code with minimum distance $D-L+1$, it is easy to verify that the joint privacy condition is satisfied. 
By the properties of MDS codes~\cite{R2006}, an $[n,k]$ MDS code (with minimum distance $n-k+1$) satisfies the following condition: 
for any ${n-k+1\leq d\leq n}$ and any $d$-subset $\mathcal{I}\subseteq \{1,\dots,n\}$, the code space contains a unique $(d-n+k)$-dimensional subspace on the coordinates indexed by $\mathcal{I}$, and any basis of this subspace forms an MDS matrix. 
This implies that the row space of the matrix $\mathrm{G}$ contains a unique $L$-dimensional subspace on every $D$-subset of coordinates, and each of these subspaces 
is equally likely to be the subspace spanned by the user's demand, from the server's perspective. 
Thus, given the query (i.e., the matrix $\mathrm{G}$), every $D$-subset of message indices is equally likely to be the support index set of the demand. 

The recoverability follows readily from the construction. 
Let ${\mathrm{U} \triangleq [\mathrm{u}_{1},\dots,\mathrm{u}_{L}]^{\mathsf{T}}}$, where $\mathrm{u}_{l}$ is a column-vector of length $K$ such that the vector $\mathrm{u}_{l}$ restricted to its components indexed by $\mathrm{W}$ is the vector $\mathrm{v}_l$, and the rest of the components of the vector $\mathrm{u}_{l}$ are all zero. 
Note that $\mathrm{V}$ is the submatrix of $\mathrm{U}$ formed by columns indexed by $\mathrm{W}$. 
We need to show that 
the rows of $\mathrm{U}$ are $L$ codewords of the code generated by $\mathrm{G}$. 
Since $\mathrm{H}$ is the parity-check matrix of the code generated by $\mathrm{G}$, this is equivalent to showing that $\mathrm{U}\mathrm{H}^{\mathsf{T}}$ is an all-zero matrix. 
Firstly, the submatrix of $\mathrm{U}\mathrm{H}^{\mathsf{T}}$ restricted to columns indexed by $\mathrm{W}$ is given by $\mathrm{V}\Lambda^{\mathsf{T}}$, and $\mathrm{V}\Lambda^{\mathsf{T}}$ is an all-zero matrix because $\Lambda$ is the parity-check matrix of the code generated by $\mathrm{V}$. 
Secondly, the submatrix of $\mathrm{U}\mathrm{H}^{\mathsf{T}}$ formed by the columns indexed by $\mathcal{K}\setminus \mathrm{W}$ is an all-zero matrix because the submatrix of $\mathrm{U}$ restricted to these columns is an all-zero matrix. Thus, $\mathrm{U}\mathrm{H}^{\mathsf{T}}$ is an all-zero matrix.
\end{proof}


\bibliographystyle{IEEEtran}
\bibliography{PIR_PC_Refs}

\end{document}